\documentclass{llncs}

\usepackage{amssymb}
\usepackage{amsfonts}
\usepackage[colorlinks,citecolor=blue,urlcolor=blue,breaklinks=true]{hyperref} 
\usepackage{tikz}
\usepackage{ifthen}
\usepackage[cmex10]{amsmath}
\usepackage{arydshln}

\usepackage{makeidx}  
 \mathchardef\ordinarycolon\mathcode`\:
   \mathcode`\:=\string"8000
   \begingroup \catcode`\:=\active
     \gdef:{\mathrel{\mathop\ordinarycolon}}
   \endgroup

\begin{document}
\frontmatter          
\pagestyle{headings}  
%
%
\mainmatter              
\title{Efficient Two-Stage Group Testing Algorithms for DNA Screening}
\titlerunning{Efficient Two-Stage Group Testing Algorithms for DNA Screening}  
%
\author{Michael Huber\thanks{The author gratefully acknowledges support of his work by the Deutsche Forschungsgemeinschaft (DFG) via a Heisenberg grant (Hu954/4) and a Heinz Maier-Leibnitz Prize grant (Hu954/5).}}
\authorrunning{Michael Huber}   
%
%
\institute{Wilhelm-Schickard-Institute for Computer Science\\
University of Tuebingen \\
Sand~13, 72076 Tuebingen, Germany.\\
\email{michael.huber@uni-tuebingen.de}}

\maketitle              

\begin{abstract}
Group testing algorithms are very useful tools for DNA \linebreak library screening. Building on recent work by  Levenshtein (2003) and Tonchev (2008), we construct in this paper new infinite classes of combinatorial structures, the existence of which are essential for attaining the minimum number of individual tests at the second stage of a two-stage disjunctive testing algorithm.
\end{abstract}


\section{Introduction}

With the completion of genome sequencing projects such as the Human Genome Project, efficient screening of DNA clones in very large genome sequence data-\linebreak bases has become an important issue pertaining to the study of gene functions. Very useful tools for DNA library screening are \emph{group testing algorithms}. The general group testing problem (cf.~\cite{Du99}) can be basically stated as follows: a large population $X$ of $v$ items that contains a small set of \emph{defective} (or \emph{positive}) items shall be tested in order to identify the defective items efficiently. For this, the items are pooled together for testing. The \emph{group test} reports ``yes'' if for a subset $S \subseteq X$ one or more defective items have been found, and ``no'' otherwise. Using a number of group tests, the task of determining which items are defective shall be accomplished. Various objectives could be considered for group testing, e.g., minimizing the number of group tests, limiting the number of pools or pool sizes, or tolerating a few errors. In what follows, we will focus on the first issue.

Of particular practical importance in DNA library screening are one- or two-stage group testing procedures (cf.~\cite[p.\,371]{Knill95}):

\begin{quote}\small

``[...] The technicians who implement the pooling strategies generally dislike even the 3-stage strategies that are often used. Thus the most commonly used strategies for pooling libraries of clones rely on a fixed but reasonably small set on non-singleton pools. The pools are either tested all at once or in a small number of stages (usually at most 2) where the previous stage determines which pools to test in the next stage. The potential positives are then inferred and confirmed by testing of individual clones [...].''

\end{quote}

\emph{Disjunctive} testing relies on Boolean operations. It aims to find the set of defective items by reconstructing its binary $(0,1)$-incidence vector $\mathbf{x} =(x_1,\ldots,x_v)$, where $x_i=1$ if the $i$th item is defective, and $x_i=0$ otherwise.
Levenshtein~\cite{Lev03} (cf. also~\cite{Ton08}) has employed a two-stage disjunctive testing algorithm in order to reconstruct the vector $\mathbf{x}$: At Stage~1, disjunctive tests are conducted which are determined by the rows of a binary matrix that is comparable to a parity-check matrix of a binary linear code. After determining what items are positive, negative or unresolved, individual tests are performed at Stage~2 in order to determine which of the remaining unresolved items are positive or negative.

Particularly important with respect to the research objectives in this paper, Levenshtein derived a combinatorial lower bound on the minimum number of individual tests at Stage~2. He showed that this bound is met with equality if and only if a Steiner $t$-design exists which has the additional property that the blocks have two sizes differing by one (i.e., $k$ and $k+1$). Relying on this result, Tonchev~\cite{Ton08} gave a straightforward construction method for such designs: Suppose that $\mathcal{D}=(X,\mathcal{B})$ is a Steiner $t$-$(v,k,1)$ design that contains a Steiner  $(t-1)$-$(v,k,1)$ subdesign $\mathcal{D}^\prime=(X,\mathcal{B}^\prime)$, where $\mathcal{B}^\prime \subseteq \mathcal{B}$. Then, the blocks of $\mathcal{D}^\prime$, each extended with one new point $x \notin X$, together with the blocks of $\mathcal{D}$ that do not belong to $\mathcal{D}^\prime$, form a Steiner \mbox{$t$-$(v+1,\{k,k+1\},1)$} design. Relying on specific balanced incomplete block designs (BIBDs), he constructed two infinite classes of such designs: A Steiner \mbox{$2$-$(q^e+1,\{q,q+1\},1)$} design exists for every prime power $q$ and every positive integer $e \geq 2$, and a Steiner $2$-$(6a+4,\{3,4\},1)$ design for every positive integer  $a$, based on resolvable BIBDs from affine geometries and Kirkman triple systems, respectively. 
Moreover, he constructed two infinite classes derived from certain Steiner quadruple systems.

In this paper, we build on the work by Levenshtein and Tonchev and  construct several further infinite classes of Steiner designs with the desired additional property. Our constructions involve, inter alia, resolvable BIBDs and cyclically resolvable BIBDs. As a result, we obtain efficient two-stage disjunctive group testing algorithms suited for DNA library screening. 

The paper is organized as follows: Combinatorial tools and structures which are important for our further purposes are provided in Section~\ref{tools}. Section~\ref{Known} presents a short overview of the previous combinatorial approaches and constructions by Levenshtein and Tonchev. Section~\ref{new} is devoted to our new combinatorial constructions. The paper concludes in Section~\ref{Concl}. 


\section{Combinatorial Structures and Tools}\label{tools}

Let $X$ be a set of $v$ elements and $\mathcal{B}$ a collection of \mbox{$k$-subsets} of $X$. The elements of $X$ and $\mathcal{B}$ are called \emph{points} and \emph{blocks}, respectively. An ordered pair \mbox{$\mathcal{D}=(X,\mathcal{B})$} is defined to be a \mbox{\emph{$t$-$(v,k,\lambda)$ design}} if each \mbox{$t$-subset} of $X$ is contained in exactly $\lambda$ blocks. For historical reasons, a \mbox{$t$-$(v,k,\lambda)$ design} with $\lambda =1$ is called a \emph{Steiner \mbox{$t$-design}} or a \emph{Steiner system}.  Well-known examples are \emph{Steiner triple systems} ($t=2$, $k=3$) and \emph{Steiner quadruple systems} ($t=3$, $k=4$).
A \emph{$2$-design} is commonly called a \emph{balanced incomplete block design}, and denoted by $\mbox{BIBD}(v,k,\lambda)$. 
It can be easily seen that in a \mbox{$t$-$(v,k,\lambda)$ design} each point is contained in the same number $r$ of blocks, and for the total number $b$ of blocks, the parameters of a \mbox{$t$-$(v,k,\lambda)$ design} satisfy the relations
\[bk=vr \quad \text{and} \quad r(k-1) = \lambda \frac{{v-2 \choose t-2}}{{k-2 \choose t-2}}(v-1) \quad \text{for} \quad t \geq 2.\]

\begin{example}\label{STS9}
Take as point-set \[X=\{1,2,3,4,5,6,7,8,9\}\] and as block-set
\[\mathcal{B}=\{\{1,2,3\},\{4,5,6\},\{7,8,9\},\{1,4,7\},\{2,5,8\},\{ 3,6,9\},\]
\[ \qquad \{1,5,9\},\{2,6,7\},\{3,4,8\},\{1,6,8\},\{2,4,9\},\{3,5,7\}\}.\]
This gives a $\mbox{BIBD}(9,3,1)$, i.e. the unique affine plane of order $3$. It can be constructed
as illustrated in Figure~\ref{AG23}.
\end{example}

\begin{figure}[htp]
\centering
\begin{tikzpicture}[scale=0.7,very thick]

\draw[black]

(0,0) -- (0,2)

(0,0) -- (2,0)

(2,0) -- (2,2)

 (0,2) -- (2,2)

 (0,1) -- (2,1)

(1,0) -- (1,2)

(0,0) -- (2,2)

(2,0) -- (0,2)

(0,1) -- (1,0)

(2,1) -- (1,2)

(0,1) -- (1,2)

(1,0) -- (2,1)

(2,1) .. controls (4,3) and (2,4) .. (0,2)

(1,2) .. controls (-1,4) and (-2,2) .. (0,0)

(0,1) .. controls (-2,-1) and (0,-2) .. (2,0)

(1,0) .. controls (3,-2) and (4,0) .. (2,2);

\filldraw [draw=red!100,fill=red!100]

(0,0) circle (2pt) (0,-0.3) node {1}

(0,1) circle (2pt) (-0.25,1) node {4}

(0,2) circle(2pt) (0,2.3) node {7}

(1,0) circle (2pt) (1,-0.3) node {2}

(1,1) circle (2pt) (1.3,1.12) node {5}

(1,2) circle(2pt) (1,2.3) node {8}

(2,0) circle (2pt) (2,-0.3) node {3}

(2,1) circle(2pt) (2.25,1) node {6}

(2,2) circle (2pt) (2,2.3) node {9};

\end{tikzpicture}
\caption{A $\mbox{BIBD}(9,3,1)$.}\label{AG23}
\end{figure}

In this paper, we primarily focus on BIBDs. Let $(X,\mathcal{B})$ be a $\mbox{BIBD}(v,k,\lambda)$, and let $\sigma$ be a permutation on $X$. For a block $B=\{b_1,\ldots,b_k\} \in \mathcal{B}$, define $B^\sigma :=\{b^\sigma_1,\ldots,b^\sigma_k\}$. If $\mathcal{B}^\sigma:=\{B^\sigma : B \in \mathcal{B}\}=\mathcal{B}$, then $\sigma$ is called an \emph{automorphism} of $(X,\mathcal{B})$. If there exists an automorphism $\sigma$ of order $v$, then the BIBD is called \emph{cyclic}. In this case, the point-set $X$ can be identified with $\mathbb{Z}_v$, the set of integers modulo $v$, and $\sigma$ can be represented by $\sigma : i \rightarrow i + 1$ (mod $v$).

For a block $B=\{b_1, \ldots, b_k\}$ in a cyclic $\mbox{BIBD}(v,k,\lambda)$, the set $B + i := \{b_1+i$ (mod $v), \ldots, b_k+i$ (mod $v)\}$ for $i \in \mathbb{Z}_v$ is called a \emph{translate} of $B$, and the set of all distinct translates of $B$ is called the \emph{orbit} containing $B$. If the length of an orbit is $v$, then the orbit is said to be \emph{full}, otherwise \emph{short}. A block chosen arbitrarily from an orbit is called a \emph{base block} (or \emph{starter block}). If $k$ divides $v$, then the orbit containing the block
\[B=\bigg\{0, \frac{v}{k}, 2 \frac{v}{k}, \ldots ,(k-1)\frac{v}{k}\bigg\}\]
is called a \emph{regular short orbit}. For a cyclic $\mbox{BIBD}(v,k,1)$ to exist, a necessary condition is $v \equiv 1$ or $k$ (mod $k(k-1)$). When $v \equiv 1$ (mod $k(k-1)$) all orbits are full, whereas if $v \equiv k$ (mod $k(k-1)$) one orbit is the regular short orbit and the remaining orbits are full.

A BIBD is said to be \emph{resolvable}, and denoted by $\mbox{RBIBD}(v,k,\lambda)$, if the block-set $\mathcal{B}$ can be partitioned into classes $\mathcal{R}_1,\ldots , \mathcal{R}_r$ such that every point of $X$ is contained in exactly one block of each class. The classes $\mathcal{R}_i$ are called \emph{resolution} (or \emph{parallel}) \emph{classes}. 
A simple example is as follows.

\begin{example}
An $\mbox{RBIBD}(9,3,1)$. Each row is a resolution class.

\[\begin{array}{|c|ccc|}
\hline
\mathcal{R}_1 & \{1,2,3\} & \{4,5,6\} & \{7,8,9\}\\
\mathcal{R}_2 & \{1,4,7\} & \{2,5,8\} & \{3,6,9\}\\
\mathcal{R}_3 & \{1,5,9\} & \{2,6,7\} & \{3,4,8\}\\
\mathcal{R}_4 & \{1,6,8\} & \{2,4,9\} & \{3,5,7\}\\
\hline
\end{array}\]
\end{example}

Generally, an $\mbox{RBIBD}(k^2,k,1)$ is equivalent to an affine plane of order $k$. An $\mbox{RBIBD}(v,3,1)$ is called a \emph{Kirkman triple system}. Necessary conditions for the existence of an $\mbox{RBIBD}(v,k,\lambda)$ are $\lambda(v-1) \equiv 0$ (mod $(k-1)$) and $v \equiv 0$ (mod $k$).

If $\mathcal{R}_i$ is a resolution class, define $\mathcal{R}^\sigma_i:=\{B^\sigma : B \in \mathcal{R}_i\}$. An RBIBD is called \emph{cyclically resolvable}, and denoted by $\mbox{CRBIBD}(v,k,\lambda)$, if it has a non-trivial automorphism $\sigma$ of order $v$ that preserves its resolution $\{\mathcal{R}_1,\ldots \mathcal{R}_r\}$, i.e., $\{\mathcal{R}^\sigma_1,\ldots \mathcal{R}^\sigma_r\}= \{\mathcal{R}_1,\ldots \mathcal{R}_r\}$ holds.
An example is as follows (cf.~\cite{Gen97}).

\begin{example}\label{example_CRB21}
A $\mbox{CRBIBD}(21,3,1)$ is given in Table~\ref{table_CRB21}. The base blocks are $\{1,4,16\}$, $\{19,20,3\}$, $\{1,11,19\}$, and $\{0,7,14\}$.
There are three full orbits and one regular short orbit. Each row is a resolution class. One orbit of resolution classes is $\{\mathcal{R}_0,\ldots,\mathcal{R}_6\}$, and another orbit is $\{\mathcal{R}^\prime_0, \mathcal{R}^\prime_1, \mathcal{R}^\prime_2\}$.
\end{example}

\begin{table*}[!t] 
\renewcommand{\arraystretch}{1.3}
\caption{Example of a $\mbox{CRBIBD}(21,3,1)$.}\label{table_CRB21}
\begin{center}
\begin{tabular}{|c|lll|}
\hline
$\mathcal{R}_0$ & $\{1,4,16\}$ $\{8,11,2\}$ $\{15,18,9\}$ &  $\{19,20,3\}$ $\{5,6,10\}$ $\{12,13,17\}$ & $\{0,7,14\}$\\
$\mathcal{R}_1$ & $\{2,5,17\}$ $\{9,12,3\}$ $\{16,19,10\}$ & $\{20,0,4\}$ $\{6,7,11\}$ $\{13,14,18\}$ & $\{1,8,15\}$\\
$\mathcal{R}_2$ & $\{3,6,18\}$ $\{10,13,4\}$ $\{17,20,11\}$ & $\{0,1,5\}$ $\{7,8,12\}$ $\{14,15,19\}$ & $\{2,9,16\}$\\
$\mathcal{R}_3$ & $\{4,7,19\}$ $\{11,14,5\}$ $\{18,0,12\}$ & $\{1,2,6\}$ $\{8,9,13\}$ $\{15,16,20\}$ & $\{3,10,17\}$\\
$\mathcal{R}_4$ & $\{5,8,20\}$ $\{12,15,6\}$ $\{19,1,13\}$ & $\{2,3,7\}$ $\{9,10,14\}$ $\{16,17,0\}$ & $\{4,11,18\}$\\
$\mathcal{R}_5$ & $\{6,9,0\}$ $\{13,16,7\}$ $\{20,2,14\}$ & $\{3,4,8\}$ $\{10,11,15\}$ $\{17,18,1\}$ & $\{5,12,19\}$\\
$\mathcal{R}_6$ & $\{7,10,1\}$ $\{14,17,8\}$ $\{0,3,15\}$ & $\{4,5,9\}$ $\{11,12,16\}$ $\{18,19,2\}$ & $\{6,13,20\}$\\
\hline
$\mathcal{R}^\prime_0$ & $\{1,11,9\}$ $\{4,14,12\}$ $\{7,17,15\}$ & $\{10,20,18\}$ $\{13,2,0\}$ $\{16,5,3\}$ & $\{19,8,6\}$\\
$\mathcal{R}^\prime_1$ & $\{2,12,10\}$ $\{5,15,13\}$ $\{8,18,16\}$ & $\{11,0,19\}$ $\{14,3,1\}$ $\{17,6,4\}$ & $\{20,9,7\}$\\
$\mathcal{R}^\prime_2$ & $\{3,13,11\}$ $\{6,16,14\}$ $\{9,19,17\}$ & $\{12,1,20\}$ $\{15,4,2\}$ $\{18,7,5\}$ & $\{0,10,8\}$\\
\hline
\end{tabular}
\end{center}
\end{table*}

Mishima and Jimbo~\cite{Jim97} classified $\mbox{CRBIBD}(v,k,1)$s into three types, according to their relation with cyclic quasiframes, cyclic semiframes, or cyclically resolvable group divisible designs. They can only exist when $v \equiv k$ (mod $k(k-1)$).

In a cyclic $\mbox{BIBD}(v,k,1)$, we can define a multiset $\Delta B :=\{b_i - b_j : i,j = 1, \ldots, k; i \neq j \}$ for a base block $B=\{b_1,\ldots,b_k\}$. Let $\{B_i\}_{i \in I}$, for some index set $I$, be all the base blocks of full orbits. If $v \equiv 1$ (mod $k(k-1)$), then clearly
\[\bigcup_{i \in I} \Delta B_i = \mathbb{Z}_v\setminus\{0\}.\]
The family of base blocks $\{B_i\}_{i \in I}$ is then called a \emph{(cyclic) difference family} in $\mathbb{Z}_v$, denoted by $\mbox{CDF}(v,k,1)$. 

Let $k$ be an odd positive integer, and $p \equiv 1$ (mod $k(k-1)$) a prime. A $\mbox{CDF}(p,k,1)$ is said to be \emph{radical}, and denoted by $\mbox{RDF}(p,k,1)$, if each base block is a coset of the $k$-th roots of unity in $\mathbb{Z}_p$ (cf.~\cite{Bur95radical}). 
A link to CRBIBDs has been established by Genma, Mishima and Jimbo~\cite{Gen97} as follows.

\begin{theorem}\label{thm_genma}
If there is an $\mbox{RDF}(p,k,1)$ with $p$ a prime and $k$ odd, then there exists a $\mbox{CRBIBD}(pk,k,1)$.
\end{theorem}

The notion of resolvability holds in the same way for $t$-$(v,k,\lambda)$ designs with $t\geq2$. Moreover, a Steiner quadruple system $3$-$(v,4,1)$ is called \emph{2-resolvable} if its block-set can be partitioned into disjoint Steiner $2$-$(v,4,1)$ designs. 

For encyclopedic references on combinatorial designs, we refer the reader to~\cite{BJL1999,crc06}.
A comprehensive book on RBIBDs and related designs is~\cite{Fur96}. Highly regular designs are treated in the monograph~\cite{Hu2008}. A recent survey on various connections between error-correcting codes and algebraic combinatorics is given in~\cite{Hu2009}. For an overview of numerous applications of combinatorial designs in computer and communication sciences, see, e. g.,~\cite{Col89,Col99,Hu2010}.


\section{Combinatorial Approaches and Constructions by Levenshtein and Tonchev}\label{Known}

This section presents a short overview of recent combinatorial approaches and constructions by Levenshtein~\cite{Lev03}  and Tonchev~\cite{Ton08}. 

\subsection{Two-Stage Disjunctive Group Testing Algorithm}

\emph{Disjunctive} group testing relies on Boolean operations in order to solve the problem of reconstructing an unknown binary vector $\mathbf{x}$ of length $v$  using the pool testing procedure~\cite{Du99}. Particularly, Levenshtein~\cite{Lev03} (see also~\cite{Ton08}) has employed a two-stage disjunctive testing algorithm to reconstruct the vector $\mathbf{x}=(x_1,\ldots,x_v)$: At Stage~1, disjunctive tests are conducted which are determined by the rows of a binary $u \times v$ matrix $H=(h_{i,j})$ that is comparable to a parity-check matrix of a binary linear code. A \emph{syndrome} $\mathbf{s} =(s_1,\ldots,s_u)$ is calculated, where $s_i$ is defined by
\[s_i=\bigvee_{j=1}^v x_j \, \& \, h_{i,j}, \quad i=1,\ldots,u,\]
where $\vee$ and $\&$ denote the logical operations of disjunction and conjunction.
The system of $u$ logical equations with $v$ Boolean variables for reconstructing the vector $\mathbf{x}=(x_1,\ldots,x_v)$ does not have a unique solution in general. After determining what items are positive, negative or unresolved, individual tests are performed at Stage~2 in order to determine which of the remaining unresolved items are positive or negative.

\subsection{Minimum Number of Individual Tests}

Let $X(v)$ be the set of all $2^v$ subsets of the set $X=\{1, 2, \ldots,  v\}$ and $X_t(v)=\{x \in X(v): \left| x \right|=t\}$.
For a fixed $t$ ($1 \leq t \leq v$) consider a covering operator $F : X_t(v) \rightarrow X(v)$ such that $x \subseteq F(x)$ for any
$x \in X_t(v)$. Define  
\[\mathcal{D} = \{F(x): x \in  X_t(v)\}.\]
For any $T$, $1\leq T \leq {v \choose t}$, consider the decreasing continuous
function $g_t(T) = k + \frac{k + 1}{t}(1- \alpha)$  where $k$ and $\alpha$ are uniquely determined by the conditions
$T{k \choose t}= \alpha {v \choose t}$, $k\in \{t,\ldots,v\}$, and $1- \frac{t}{k+1} < \alpha \leq 1$.
Using averaging and linear programing,  Levenshtein~\cite{Lev03} proved the following inequality:

\begin{theorem}[Levenshtein, 2003]\label{Lev}

\[\frac{1}{{v \choose t}} \sum_{x \in X_t(v)} \left| F(x) \right| \geq g_t(\left| \mathcal{D} \right|),\]
and  the bound is met with equality if and only if $\mathcal{D}$ is a Steiner \mbox{$t$-$(v,\{k,k+1\},1)$} design.
\end{theorem}

One of the main motivations for the above result is to minimize the number of individual tests at the second stage of a two-stage disjunctive group testing algorithm under the condition that the vectors $\mathbf{x}$ are distributed with probabilities $p^{\left| x \right|} (1-p)^{v-\left| x \right|}$ where $x \in X(v)$  denotes the indices of the ones (defective items) in $\mathbf{x}$. The bound above implies that the expected number of items which remain unresolved after application in parallel of $u$ pools is not less than
\[v \sum_{t=1}^{v} {{v \choose t}} p^t (1-p)^{v-t}2^{-\frac{u}{t}}-vp.\]


\subsection{Known Infinite Classes of Combinatorial Constructions}

Tonchev~\cite{Ton08} straightforwardly gave a non-trivial construction method to obtain Steiner designs which have the additional property that the blocks have two sizes differing by one. 

\begin{proposition}[Tonchev, 2008]\label{Ton}
Suppose that $\mathcal{D}=(X,\mathcal{B})$ is a Steiner \linebreak \mbox{$t$-$(v,k,1)$} design that contains a Steiner  $(t-1)$-$(v,k,1)$ subdesign $\mathcal{D}^\prime=(X,\mathcal{B}^\prime)$, where $\mathcal{B}^\prime \subseteq \mathcal{B}$. Then, the blocks of $\mathcal{D}^\prime$, each extended with one new point $x \notin X$, together with the blocks of $\mathcal{D}$ that do not belong to $\mathcal{D}^\prime$, form a Steiner \mbox{$t$-$(v+1,\{k,k+1\},1)$} design. In particular, if there exists an $\mbox{RBIBD}(v,k,1)$, then there exists a Steiner \mbox{$2$-$(v+1,\{k,k+1\},1)$} design.
\end{proposition}

Relying on resolvable BIBDs from affine geometries and Kirkman triple systems, Tonchev derived from the above result the following infinite classes:

\begin{theorem}[Tonchev, 2008]\label{Ton2}
There exists

\begin{enumerate}

\item[$\bullet$] a Steiner \mbox{$2$-$(q^{e}+1,\{q,q+1\},1)$} design for any prime power $q$ and any positive integer $e \geq 2$,

\item[$\bullet$] a Steiner \mbox{$2$-$(6a+4,\{3,4\},1)$} design for any positive integer $a$.

\end{enumerate}
 
\end{theorem}

Based on results on 2-resolvable Steiner quadruple systems by Baker~\cite{Bak76} \& Semakov et al.~\cite{Sem73} and by Teirlinck~\cite{Teir94}, Tonchev obtained this way also two infinite classes for $t > 2$. The third class had already been constructed earlier by Tonchev~\cite{Ton96}. 

\begin{theorem}[Tonchev, 1996 \& 2008]\label{Ton3}
There is

\begin{enumerate}

\item[$\bullet$] a Steiner \mbox{$3$-$(2^{2f}+1,\{4,5\},1)$} design for any positive integer $f \geq 2$,

\item[$\bullet$] a Steiner \mbox{$3$-$(2 \cdot 7^e+3,\{4,5\},1)$} design for any positive integer $e$,

\item[$\bullet$] a Steiner \mbox{$4$-$(4^{e}+1,\{5,6\},1)$} design for any positive integer $e \geq 2$.

\end{enumerate}

\end{theorem}


\section{New Infinite Classes of Combinatorial Constructions}\label{new}

We present several constructions of new infinite families of Steiner designs having the desired additional property that the blocks have two sizes differing by one. Our constructions involve, inter alia, resolvable BIBDs and cyclically resolvable BIBDs. As a result, we obtain efficient two-stage disjunctive group testing algorithms suited for DNA library screening. 


\subsection{CRBIBD-Constructions}

We obtain the following result:

\begin{theorem}\label{thm_CRBIBD-Constructions}
Let $p$ be a prime. Then there exists a Steiner \mbox{$2$-$(pk+1,\{k,k+1\},1)$} design for the following cases:

\begin{enumerate}

\item[\em(1)] $(k,p)= (3,6a+1)$ for any positive integer $a$,

\item[\em(2)] $(k,p)= (4,12a+1)$ for any odd positive integer $a$,

\item[\em(3a)] $(k,p)= (5,20a+1)$ for any positive integer $a$ such that $p<10^3$, and furthermore

\item[\em(3b)] $(k,p)= (5,20a+1)$ for any positive integer $a$ satisfying the condition stated in~(ii) in the proof,

\item[\em(4)] $(k,p)= (7,42a+1)$ for any positive integer $a$ satisfying the condition stated in~(iii) in the proof,

\item[\em(5)] $(k,p)= (9,p)$ for the values of $p \equiv 1$ $($\emph{mod} $72)  < 10^4$ given in Table~\ref{table_RDF}.

\end{enumerate}

Moreover, there exists a  Steiner \mbox{$2$-$(qk+1,\{k,k+1\},1)$} design for the following cases:

\begin{enumerate}

\item[\em(6)] $(k,q)$ for $k=3,5,7$, or $9$, and $q$ is a product of primes of the form $p \equiv 1$ $($\emph{mod} $k(k-1))$ as in the cases above,

\item[\em(7)] $(k,q)= (4,q)$ and $q$ is a product of primes of the form $p=12a +1$ with $a$ odd.

\end{enumerate}

\end{theorem}

\begin{proof}
The constructions are based on the existence of a $\mbox{CRBIBD}(pk,k,1)$ in conjunction with Proposition~\ref{Ton}.
We first assume that $k$ is odd. Then the following infinite ((i)-(iii)) and finite ((iv)) families of radical difference families exist (cf.~\cite{Bur95radical} and the references therein;~\cite{crc06}):

\begin{enumerate}

\item[(i)] An $\mbox{RDF}(p,3,1)$ exists for all primes $p \equiv 1$ (mod $6$).

\item[(ii)] Let $p=20a+1$ be a prime, let $2^e$ be the largest power of $2$ dividing $a$ and let $\varepsilon$ be a $5$-th primitive root of unity in $\mathbb{Z}_p$. Then an $\mbox{RDF}(p,5,1)$ exists if and only if $\varepsilon +1$ is not a $2^{e+1}$-th power in $\mathbb{Z}_p$, or equivalently $(11+5\sqrt{5})/2$ is not a $2^{e+1}$-th power in $\mathbb{Z}_p$.

\item[(iii)] Let $p=42a+1$ be a prime and let $\varepsilon$ be a $7$-th primitive root of unity in $\mathbb{Z}_p$. Then an $\mbox{RDF}(p,7,1)$ exists if and only if there exists an integer $f$ such that $3^f$ divides $a$ and $\varepsilon +1, \varepsilon^2 + \varepsilon +1,\frac{\varepsilon^2 + \varepsilon +1}{\varepsilon +1}$ are $3^{f}$-th powers but not $3^{f+1}$-th powers in $\mathbb{Z}_p$.

\item[(iv)] An $\mbox{RDF}(p,9,1)$ exists for all primes $p<10^4$ displayed in Table~\ref{table_RDF}.

\end{enumerate}

Theorem~\ref{thm_genma} yields the respective $\mbox{CRBIBD}(pk,k,1)$s. Moreover, in~\cite{Gen97} a recursive construction is given that implies the existence of a $\mbox{CRBIBD}(kq,k,1)$ whenever $q$ is a product of primes of the form $p \equiv 1$ (mod $k(k-1)$).
In addition, a $\mbox{CRBIBD}(5p,5,1)$ has been shown~\cite{Bur97resolvable} to exist for any prime $p \equiv 1$ (mod $20$) $< 10^3$. 

We now consider the case when $k$ is even:
In~\cite{Lam99}, a $\mbox{CRBIBD}(4p,4,1)$ is constructed for any prime $p=20a+1$, where $a$ is an odd positive integer.
Furthermore, via the above recursive construction, a $\mbox{CRBIBD}(4q,4,1)$ exists whenever $q$ is a product of primes of the form $p=12a +1$ and $a$ is odd. The result follows.\qed
\end{proof}

\begin{table}[!t] 
\renewcommand{\arraystretch}{1.3}
\caption{Existence of an $\mbox{RDF}(p,k,1)$ with $k=5$, $p<10^3$, and $k=7$ or $9$, $p < 10^4$.}\label{table_RDF}
\begin{center}
\begin{tabular}{|ccccccccc|}
  \hline
   & & & & $k=5$ & & & &\\
  \hline
   41 & 61 & 241 & 281 & 401 & 421 & 601 & 641 & 661\\
   701 & 761 & 821 & 881 &  &  &  &  & \\
   \hline
   \hline
   & & & & $k=7$ & & & & \\
  \hline
   337  & 421  &  463 & 883 & 1723 & 3067 & 3319 & 3823 & 3907 \\
   4621 & 4957 & 5167 & 5419 & 5881 & 6133 & 8233 & 8527 & 8821 \\
   9619 & 9787 & 9829 & & & & & & \\
   \hline
   \hline
   & & & & $k=9$ & & & & \\
  \hline
   73  & 1153  & 1873  & 2017 & 6481 & 7489 & 7561 &  &  \\
   \hline
\end{tabular}
\end{center}
We remark that further parameters are given in~\cite{Bur95radical} for $\mbox{RDF}(p,k,1)$s with $k=7$ or $9$ and $10^4 \leq p < 10^5$.
\end{table}

\begin{example}
Values of $p$ for which an $\mbox{RDF}(p,k,1)$ exists with $k=5$, $p<10^3$, and $k=7$ or $9$, $p < 10^4$ are displayed in Table~\ref{table_RDF} (cf.~\cite{crc06}). For example, if we take an $\mbox{RDF}(41,5,1)$, then we obtain a 
Steiner \mbox{$2$-$(206,\{5,6\},1)$} design. If we take an $\mbox{RDF}(61,5,1)$, then we obtain a Steiner \mbox{$2$-$(306,\{5,6\},1)$} design.
\end{example}


\subsection{RBIBD-Constructions}

We establish the following result:

\begin{theorem}\label{thm_RBIBD-Constrctions}
Let $v$ be a positive integer. Then there exists a Steiner \linebreak \mbox{$2$-$(v+1,\{k,k+1\},1)$} design for the following cases:

\begin{enumerate}

\item[\em(1)] $(k,v)= (3,6a+3)$ for any positive integer $a$,

\item[\em(2)] $(k,v)= (4,12a+4)$ for any positive integer $a$,

\item[\em(3)] $(k,v)= (5,20a+5)$ for any positive integer $a$ with the possible exceptions given in Table~\ref{table_RBIBD},

\item[\em(4)] $(k,v)= (8,56a+8)$ for any positive integer $a$ with the possible exceptions given in Table~\ref{table_RBIBD}.

\end{enumerate}

\end{theorem}

\begin{proof}
The constructions are based on the existence of an $\mbox{RBIBD}(v,k,1)$ in conjunction with Proposition~\ref{Ton}.
The following infinite series of resolvable balanced incomplete block designs are known (cf.~\cite{GA97,crc06} and the references therein):
\begin{enumerate}

\item[(i)] When $k=3$ and $4$, respectively, an $\mbox{RBIBD}(v,k,1)$ exists for all positive integers $v \equiv k$ (mod $k(k-1)$).

\item[(ii)] An $\mbox{RBIBD}(v,5,1)$ exists for all positive integers $v \equiv 5$ (mod $20$) with the possible exceptions given in Table~\ref{table_RBIBD}.

\item[(iii)] An $\mbox{RBIBD}(v,8,1)$ exists for all positive integers $v \equiv 8$ (mod $56$) with the possible exceptions given in Table~\ref{table_RBIBD}.

\end{enumerate}
This proves the theorem.\qed
\end{proof}

We remark that Case $(1)$ has already been covered in Theorem~\ref{Ton2}.

\begin{example}Choosing for example an $\mbox{RBIBD}(65,5,1)$, we get a Steiner \linebreak  \mbox{$2$-$(66,\{5,6\},1)$} design. 
If we choose an $\mbox{RBIBD}(105,5,1)$, then we obtain a Steiner \mbox{$2$-$(106,\{5,6\},1)$} design.
\end{example}

\begin{table}
\renewcommand{\arraystretch}{1.3}
\caption{Possible exceptions: An $\mbox{RBIBD}(v,k,1)$ with $k=5$ or $8$ is not known to exist for the following values of \mbox{$v \equiv k$ (mod $k(k-1)$)}.}\label{table_RBIBD}
\begin{center}
\begin{tabular}{|cccccccc|}
  \hline
   & & & $k=5$ & & & &\\
  \hline
   45 & 345 & 465 & 645 & & & &\\
   \hline
   \hline
   & & &  $k=8$ & & & & \\
  \hline
   176  & 624  & 736  & 1128 & 1240 & 1296 & 1408 & 1464\\
   1520 &1576  & 1744 & 2136 & 2416 & 2640 & 2920 & 2976\\
   3256 & 3312 & 3424 & 3760 & 3872 & 4264 & 4432 & 5216\\
   5720 & 5776 & 6224 & 6280 & 6448 & 6896 & 6952 & 7008\\
   7456 & 7512 & 7792 & 7848 & 8016 & 9752 & 10200 & 10704\\
   10760 & 10928 & 11040 & 11152 & 11376 & 11656 & 11712 & 11824\\
   11936 & 12216 & 12328 & 12496 & 12552 & 12720 & 12832 & 12888\\
   13000 & 13280 & 13616  & 13840 & 13896 & 14008 & 14176 & 14232\\
   21904 & 24480 &  &  & & & & \\
   \hline
\end{tabular}
\end{center}
\end{table}

\begin{theorem}\label{thm_RBIBD-constructions_prime}
If $v$ and $k$ are both powers of the same prime, then a Steiner \mbox{$2$-$(v+1,\{k,k+1\},1)$} 
exists if and only if $(v-1) \equiv 0$ $($\emph{mod} $(k-1))$ and $v \equiv 0$ $($\emph{mod} $k)$.
\end{theorem}

\begin{proof}
It has been shown in~\cite{Greig06} that, for $v$ and $k$  both powers of the same prime, the necessary conditions for the existence of an $\mbox{RBIBD}(v,k,\lambda)$ are sufficient. Hence, the result follows via Proposition~\ref{Ton} when considering an $\mbox{RBIBD}(v,k,1)$.\,$\Box$
\end{proof}


\subsection{3-Design-Constructions}

Based on further results on 2-resolvable Steiner quadruple systems described in Teirlinck~\cite{Teir94}, we obtain this way also two infinite classes for $t > 2$. 

\begin{theorem}\label{thm_3-designs-sconstructions}

There exists

\begin{enumerate}

\item[$\bullet$] a Steiner \mbox{$3$-$(2 \cdot {31}^e+3,\{4,5\},1)$} design for any positive integer $e$,

\item[$\bullet$] a Steiner \mbox{$3$-$(2 \cdot {127}^e+3,\{4,5\},1)$} design for any positive integer $e$.

\end{enumerate}

\end{theorem}

\section{Conclusion}\label{Concl}

Efficient two-stage group testing algorithms that are particular suited for DNA library screening have been investigated in this paper. The main focus has been on novel combinatorial constructions in order to minimize the number of individual tests at the second stage of a two-stage disjunctive testing algorithm. Several infinite classes of such combinatorial structures have been obtained.

\clearpage
\addtocmark[2]{Author Index} 
\renewcommand{\indexname}{Author Index}
\clearpage
\addtocmark[2]{Subject Index} 
\markboth{Subject Index}{Subject Index}
\renewcommand{\indexname}{Subject Index}
\end{document}